\documentclass[journal]{IEEEtran}

\usepackage[utf8]{inputenc}

\normalsize
\DeclareMathAlphabet\mathbfcal{OMS}{cmsy}{b}{n}

\usepackage[english]{babel}
\usepackage{amsmath,amssymb}
\usepackage{cite}
\usepackage{graphicx}
\usepackage{dsfont}
\usepackage{amsthm}
\usepackage{amssymb}
\usepackage{mdframed}
\usepackage{subfigure}
\usepackage{xspace}
\usepackage[colorinlistoftodos,textwidth=\marginparwidth,
disable
]{todonotes}
\newcommand{\todoa}[2][]{\xspace\todo[size=\tiny,color=blue!20!white,#1]{A: #2}}
\newcommand{\todoe}[2][]{\xspace\todo[size=\tiny,color=orange!20!white,#1]{E: #2}}

\newcommand{\idle}{\mathrm{i}}
\newcommand{\retx}{\mathrm{x}}
\newcommand{\new}{\mathrm{n}}
\newcommand{\sense}{\mathrm{n}}
\usepackage{algorithm}
\usepackage{algorithmicx}
\usepackage{algcompatible}
\usepackage{eqparbox}
\usepackage{enumerate}

\makeatletter
\def\blfootnote{\xdef\@thefnmark{}\@footnotetext}
\makeatother

\usepackage{lscape}
\usepackage{bbm}
\DeclareMathOperator*{\argmin}{arg\,min} 
\DeclareMathOperator*{\argmax}{arg\,max} 
\newtheorem{theorem}{Theorem}
\newtheorem{problem}{Problem}

\newtheorem*{previous theorem}{Theorem~1 of \cite{journal_paper}}

\newcommand{\Exp}[1]{\mathbb{E}\left[ #1 \right]} 
 
\newcommand{\Exppi}[1]{\mathbb{E}\left[ #1 \right]} 

\title{Learning to Minimize Age of Information over an Unreliable Channel with Energy Harvesting}

\author{
\IEEEauthorblockN{Elif Tu\u{g}\c{c}e Ceran, Deniz G{\"u}nd{\"u}z, and Andr\'as Gy\"orgy}
}

\begin{document}

\maketitle

\begin{abstract}
The time average expected age of information (AoI) is studied for status updates sent over an error-prone channel from an energy-harvesting transmitter with a finite-capacity battery. Energy cost of sensing new status updates is taken into account as well as the transmission energy cost better capturing practical systems. The optimal scheduling policy is first studied under the hybrid automatic repeat request (HARQ) protocol when the channel and energy harvesting statistics are known, and the existence of a threshold-based optimal policy is shown. For the case of unknown environments,  average-cost reinforcement-learning algorithms are proposed that learn the system parameters and the status update policy in real-time. The effectiveness of the proposed methods is demonstrated through numerical results.
\end{abstract}

\blfootnote{Part of this work is was presented at the  IEEE Conference on Computer Communications Workshops (INFOCOM WKSHPS), Paris, France, April 2019 \cite{infocom_paper}.}
\blfootnote{This work was supported in part by the European Research Council (ERC) Starting Grant BEACON (grant agreement no. 677854).
E.~T.~Ceran and D.~G\"und\"uz are with 
Imperial College London, UK (email: \texttt{\{e.ceran14, d.gunduz\}@imperial.ac.uk}. A.~Gy\"orgy is with DeepMind, 
UK (email: \texttt{agyorgy@google.com}).}


\section{Introduction}

Many status update systems, including wireless sensor networks in Internet of things (IoT) applications, are powered by scavenging energy from renewable sources (e.g., solar cells \cite{solar}, wind turbines \cite{wind}, piezoelectric generators \cite{vibration}, etc.). Harvesting energy from ambient sources provides environmentally-friendly and ubiquitous operation for remote sensing systems.  Therefore, there has been a growing interest in maximizing the timeliness of information in energy harvesting (EH) communication systems \cite{Tan2015,Yates2015,Tan2017,Arafa2017,Wu2018,Tan2018,Feng2018,ArafaPoor2018,pappas2019,hatami2020,gind2020,pappas2020}. In these systems, the staleness of the information at the receiver is measured by the age of information (AoI), defined as the time elapsed since the generation time of the most recent status update  successfully received at the receiver.




Prior works have investigated online \cite{Tan2015,Tan2017,Tan2018,pappas2019,hatami2020} as well as offline \cite{Tan2015,Arafa2017,gind2020} methods for different scenarios in order to optimize the timeliness of information under energy causality constraints in EH systems. The structure of an optimal policy is derived and heuristic algorithms are proposed in \cite{Tan2017,Tan2018,ArafaPoor2018,pappas2019} 
for a finite-size battery considering only the cost of transmissions. Until recently, literature on AoI assumed that the cost of sensing (monitoring) the status of a process is negligible compared to the cost of transmitting the status updates. However, in many practical sensing systems acquiring a new sample of the underlying process of interest also has a considerable energy cost \cite{sensing_energy, Gong2018}. The sampling/sensing cost has been taken into account in \cite{Gong2018} and \cite{Zhou2018}, where a status update system with automatic repeat request (ARQ) and an unlimited energy source is considered. In \cite{Gong2018}, closed form expressions are presented for the energy consumption and average AoI with known transmission error probability, assuming that a packet is re-transmitted until either it is received, or a prescribed maximum number of transmissions is reached.  In our previous work, we studied status-update systems under a transmission-rate constraint, or equivalently, an average power constraint \cite{wcnc_paper, pimrc_paper, journal_paper}.






In this paper, we study a status update system considering both the sensing and transmission energy costs, better capturing practical systems. \todoa{What is better here? Other works also consider transmission energy. Why is this sentence not in the previous paragraph?}\todoe{modified} Moreover, we consider an EH transmitter, which uses the energy harvested from the environment to power the sensing and communication operations. Unlike prior work \cite{Gong2018}, we consider the intermittent availability of energy and a hybrid automatic repeat request (HARQ) protocol, where the partial information obtained from previous unsuccessful transmission attempts is combined to increase the probability of successful decoding. When employing HARQ in status update systems, there is an inherent trade-off between sending a new update after a failed transmission attempt, which may result in a lower AoI, and retransmitting the failed update, which has a lower probability of error. Introducing sensing cost to the system model makes this trade-off even more challenging and interesting, as retransmissions incur no sensing cost. 

In most practical scenarios statistical information about either the energy arrival process or the channel conditions are not available a-priori, or may change over time \cite{Gunduz2014}. Previous works on EH communication systems without a-priori information on the random processes governing the system exploited reinforcement learning (RL) methods in order to maximize throughput or minimize delay \cite{Blasco2013, Ortiz2016}. In this paper, we propose RL algorithms that can adapt the status-update scheme to the unknown energy arrival process as well as the channel statistics.

Our goal will be to identify the optimal policy that can judiciously balance the AoI benefits of transmitting a new status update with its additional sensing cost and lower success probability. The optimal decisions will depend not only on the current AoI and retransmission count, but also on the battery and energy harvesting states. We consider a value-based RL algorithm, \emph{GR-learning} \cite{Gosavi2004}, a policy-based RL algorithm, \emph{finite-difference policy gradient (FDPG)} \cite{policy_gradient}, and a deep RL algorithm \emph{deep Q-network (DQN)}, and compare their performances with that of the relative value iteration (RVI) algorithm which assumes a-priori knowledge on the system characteristics. We propose threshold policy with low computational complexity \todoa{Why do we propose a suboptimal policy? Why do we use a threshold policy?}\todoe{Modified} and demonstrate that a policy gradient algorithm exploiting the structural characteristics of a threshold policy outperforms the GR-learning algorithm. 


The main contributions of this paper are outlined as follows:

\begin{itemize}
\item The average AoI is studied under energy replenishment constraints, i.e., energy causality as well as battery capacity constraints, imposed on the transmitter, which limits the energy consumption in a stochastic manner.
\item The optimal policy is shown to be stationary, deterministic, and monotone with respect to the AoI, while both retransmissions and preemption\todoa{This is not defined yet.}\todoe{I thought preemption is a standard definition, do we need to define it? Or did I misunderstood the comment?} following a failed transmission are considered. 
\item Scheduling algorithms are designed using multiple average-cost RL algorithms; in particular, a value-based RL algorithm, \emph{GR-learning} \cite{Gosavi2004}, a policy-based RL algorithm, \emph{FDPG} \cite{policy_gradient}, and a deep RL algorithm, \emph{(DQN)}, are used to learn the optimal scheduling decisions when the transmission success probabilities and energy arrival statistics are unknown. 
\item Numerical simulations are conducted in order to investigate the effects of the EH process on the average AoI. In particular, we have found that temporal correlations in EH increase the average AoI significantly.
\item  The average AoI with EH is compared with the average AoI under an average transmission constraint, and it is shown that the average AoI obtained by the EH transmitter is similar \todoa{What does "approximates" mean? Does it mean they are similar? In what sense?}\todoe{modified} to the one obtained under an average transmission constraint for a battery with unlimited capacity and zero sampling/sensing cost.
\end{itemize}

The rest of the paper is organized as follows. The system model is presented in Section~\ref{sec:system_model_eh}. The problem of minimizing the average AoI with HARQ under energy replenishment constraints is formulated as a Markov decision process (MDP) and the structure of the optimal policy is investigated in Section~\ref{sec:dynamic}. Section~\ref{sec:learning_EH} shows the application of RL algorithms to minimize the AoI in an unknown environment.  Simulation results are presented in Section~\ref{sec:results_EH}, and the paper is concluded in Section~\ref{sec:conclusion_EH}.


\section{System Model}
\label{sec:system_model_eh}
\begin{figure}
\centering
\includegraphics[scale=0.4]{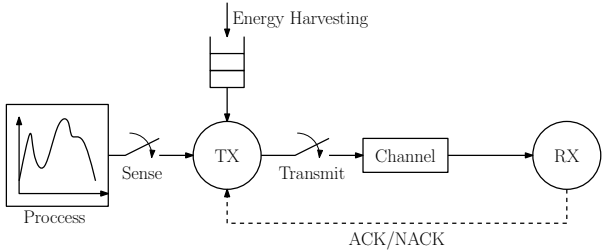}
\caption{An EH status update system over an error-prone link in the presence of ACK/NACK feedback.}
\label{fig:system_EH}
\end{figure}

We consider a time-slotted status update-system over an error-prone wireless communication link (see Figure~\ref{fig:system_EH}). The transmitter (TX) can sense the underlying time-varying process and  generate a status update in each time slot at a certain energy cost. Status updates are communicated to the receiver (RX) over a time-varying wireless channel. Each transmission attempt takes constant time, which is assumed to be equal to the duration of one time slot. 

The AoI measures the timeliness of the status information at the receiver, and is defined at any time slot $t$ as the number of time slots elapsed since the generation time $U(t)$ of the most up-to-date packet successfully decoded at the receiver. Formally, the AoI at the receiver at time $t$ is defined as $\Delta^{rx}_t\triangleq \min(t-U(t),\Delta_{max})$, where a maximum value $\Delta_{max}$ on the AoI is imposed to limit the  impact of the AoI on the performance after some level of staleness is reached.\todoa{Clarify if receiving the update in the same time slot counts to what age.}\todoe{Sorry, I couldn't understand the comment}

We assume that the channel changes randomly from one time slot to the next in an independent and identically distributed (i.i.d.) fashion, and the instantaneous channel state information is available only at the receiver. We further assume the availability of an error- and delay-free single-bit feedback from the receiver to the transmitter for each transmission attempt. Successful reception of the status update at the end of time slot $t$ is acknowledged by an ACK signal (denoted by $K_t=1$), while a NACK signal is sent in case of a failure (denoted by $K_t=0$).

There are three possible actions the transmitter can take in each time slot $t$ (the transmitter's action is denoted by $A_t$). It can either sample and transmit a new status update ($A_t=\new$), remain idle ($A_t=\idle $), or retransmit the last failed status update ($A_t= \retx$). If an ACK is received at the transmitter, we can restrict the action space to $\{\idle,\new\}$ as retransmitting an already decoded status update is strictly suboptimal. Also note that, even though the transmitter can just sense and generate a new update but not transmit it, this would be suboptimal, so we do not consider such an action separately.


We consider the HARQ protocol: that is, the received signals from previous transmission attempts for the same packet are combined for decoding. The probability of error using $r$ retransmissions, denoted by $g(r)$, depends on $r$ and the particular HARQ scheme used for combining multiple transmission attempts (an empirical method to estimate $g(r)$ is presented in \cite{harq2003}). As in any reasonable HARQ strategy, we assume that $g(r)$ is non-zero for any $r$ and is non-increasing in the number of retransmissions $r$; that is, $g(r_1) \geq g(r_2) > 0 $ for all $r_1 \leq r_2$.  We also assume that the transmissions are successful with a positive probability $g(r)< 1$ for all $r$. Standard HARQ methods only combine information from a finite maximum number of retransmissions \cite{IEEEstandard}. Accordingly, we consider a truncated retransmission count of a status update, denoted by $R_t$ for the status update transmitted at time $t$, where $R_t \in \{0,\ldots,R_{max}\}$; that is, the receiver can combine information from the last $R_{max}$ retransmissions at most. We set $R_0=0$ so that there is no previously transmitted packet at the transmitter at time $t=0$.


At the end of each time slot $t$, a random amount of energy is harvested and stored in a rechargeable battery at the transmitter, denoted by $E_t \in\mathcal{E} \triangleq \{0,1,\ldots,E_{max}\}$, following  a first-order discrete-time Markov model, characterized by the stationary transition  probabilities $p_E(e_1|e_2)$, defined as $p_E(e_1|e_2)\triangleq Pr(E_{t+1}=e_2|E_t=e_1)$, $\forall t$ and $\forall e_1,e_2 \in \mathcal{E}$. It is also assumed that $p_E(0|e)>0$, $\forall e \in \mathcal{E}$. Harvested energy is first stored in a rechargeable battery with a limited capacity of $B_{\mathrm{max}}$ energy units. The energy consumption for status sensing is denoted by $E^{\mathrm{s}}\in \mathbb{Z}^{+}$, while the energy consumption for a transmission attempt is denoted by $E^{\mathrm{tx}} \in \mathbb{Z}^{+}$. 





The battery state at the beginning of time slot $t$, denoted by $B_t$, can be written as follows:
\begin{align} \label{eq:causality1}
B_{t+1}=\min(B_t+E_t-(E^{\mathrm{s}}+E^{\mathrm{tx}})\mathbbm{1}[A_t=\sense]\nonumber\\-E^{\mathrm{tx}}\mathbbm{1}[A_t=\retx],B_{\mathrm{max}}), 
\end{align}
and the energy causality constraints are given as: 
\begin{equation}
(E^{\mathrm{s}}+E^{\mathrm{tx}})\mathbbm{1}[A_t=\sense]+E^{\mathrm{tx}}\mathbbm{1}[A_t=\retx] \leq B_t,\label{eq:causality2}
\end{equation}
where the indicator function $\mathbbm{1}[C]$ is equal to $1$ if event $C$ holds, and zero otherwise. \eqref{eq:causality1} implies that the battery overflows if energy is harvested when the battery is full, while \eqref{eq:causality2} imposes that the energy consumed by sensing or transmission operations at time slot $t$ is limited by the energy $B_{t}$ available in the battery at the beginning of that time slot. We set $B_0 = 0$ so that the battery is empty at time $t = 0$.


Let $\Delta^{tx}_t$ denote the number of time slots elapsed since the generation of the most recently sensed status update at the transmitter side, while $\Delta^{rx}_t$ denote the AoI of the most recently received status update at the receiver side. $\Delta^{tx}_t$ resets to $1$ if a new status update is generated at time slot $t-1$, and increases by one (up to $\Delta_{max}$) otherwise, i.e., 
\begin{align*}
\Delta^{tx}_{t+1}=
\begin{cases}
1 &\textrm{ if }  A_t=\sense; \\
\min(\Delta^{tx}_t+1,\Delta_{max}) &\textrm{ otherwise. }
\end{cases}
\end{align*}
On the other hand, the AoI at the receiver side evolves as follows:
\begin{align*}\Delta^{rx}_{t+1}\!=\!
\begin{cases}
\min(\!\Delta^{rx}_t\!+\!1,\!\Delta_{max}\!) &\textrm{if } A_t\!=\!\idle \textrm{ or } K_t\!=\!0; \\
1 &\textrm{if } A_t\!=\!\sense \textrm{ and } K_t\!=\!1; \\
\min(\!\Delta^{tx}_t\!+\!1,\!\Delta_{max}\!) &\textrm{if } A_t\!=\!\retx \textrm{ and } K_t\!=\!1.
\end{cases}
\end{align*}
Note that once the AoI at the receiver is at least as large as at the transmitter, this relationship holds forever; thus it is enough to consider cases when $\Delta^{rx}_t \ge \Delta^{tx}_t$. 

To determine the success probability of a transmission, we need to keep track of the number of current retransmissions. The number of retransmissions is zero for a newly sensed and generated status update and increases up to $R_{max}$ as we keep retransmitting the same packet. It is easy to see that retransmitting when $\Delta^{tx}_{t+1}=\Delta_{max}$ is suboptimal, therefore we explicitly exclude this action by setting the retransmission count to 0 in this case. Also, it is suboptimal to generate a new update and retransmit the old one, thus whenever a new status update is generated, the previous update at the transmitter is dropped and cannot be retransmitted. Thus, the evolution of the retransmission count is given as follows:  
\begin{align*}
R_{t+1}\!=\!
\begin{cases}
0 &\textrm{ if }  K_t=1 \\
& \textrm{ or } \Delta^{tx}_{t+1}=\Delta_{max}; \\
1 &\textrm{ if } A_t\!=\!\sense \textrm{ and } K_t\!=\!0; \\
R_t &\textrm{ if }  A_t\!=\!\idle \\
& \textrm{ and } \Delta^{tx}_{t+1}\neq\Delta_{max}; \\
\min(R_t\!+\!1\!,\!R_{max}) &\textrm{ if } A_t\!=\!\retx, K_t\!=\!0 \\
&  \textrm{ and } \Delta^{tx}_{t+1}\neq\Delta_{max}.
\end{cases}
\end{align*}

The state of the system is formed by five components $S_t=(E_{t},B_t,\Delta^{rx}_t,\Delta^{tx}_t,R_t)$. In each time slot, the transmitter knows the state, and takes action from the set $\mathcal{A}=\{\idle,\new,\retx\}$. The goal is to find a policy $\pi$ which minimizes the expected average AoI at the receiver over an infinite time horizon, which is given by:
\begin{problem}
\begin{align}
J^*\triangleq \min_{\pi}\lim_{T\rightarrow \infty }\frac{1}{T+1}\Exp{\sum_{t=0}^T{\Delta^{rx}_t}} \label{eq:problem}\\
\textrm{subject to } \eqref{eq:causality1} \textrm{ and } \eqref{eq:causality2} \nonumber.
\end{align}
\label{problem_EH}
\end{problem}

In \cite{journal_paper}, we have considered status updates with HARQ under an average power constraint. In that case, it is possible to show that it is suboptimal to retransmit a failed update after an idle period. Restricting the actions of the transmitter accordingly, the AoI at the receiver after a successful transmission event is equal to the number of retransmissions of the corresponding update. Therefore in addition to the AoI at the receiver, we only need to track the retransmission count. However, in the current scenario, retransmissions of a status update can be interrupted due to energy outages, which means that we also need to keep track of the AoI at the transmitter side (hence we need to have both $\Delta^{rx}_t$ and $\Delta^{tx}_t$ in the state of the system).

\section{MDP Formulation}
\label{sec:dynamic}

It is easy to see that Problem~\ref{problem_EH} can be formulated as an average-cost finite-state MDP: An MDP is defined by the quadruple  $\big(\mathcal{S}, \mathcal{A}, P, c\big)$ \cite{Puterman_book}: The finite set of states $\mathcal{S}$ is defined as $\mathcal{S}=\{s=(e,b,\delta^{rx},\delta^{tx},r) : e \in \mathcal{E}, b \in \{0,\ldots,B_{\mathrm{max}}\}, \delta^{rx},\delta^{tx} \in  \{1,\ldots,\Delta_{max}\}, r \in \{0,\ldots,R_{max}\}, \Delta^{rx} \ge \Delta^{tx}\}$,  while the finite set of actions $\mathcal{A}=\{\idle,\sense,\retx\}$ is already defined. Note that action $\retx$ cannot be taken in states with retransmission count $r=0$. $P$ refers to the transition probabilities, where  $P(s'|s,a) = \Pr(S_{t+1}=s' \mid S_t = s, A_t=a)$ is the probability that action $a$  in state $s$  at time $t$  will lead to state $s'$ at time $t+1$, which is characterized by the EH statistics and channel error probabilities. The cost function $c: \mathcal{S} \times \mathcal{A} \rightarrow \mathbbm{Z}$, is the AoI at the receiver, and is defined as $c(s,a)=\delta^{rx}$ for any $s\in \mathcal{S}$, $a \in \mathcal{A}$, independent of the action taken, where $\delta^{rx}$ is the component of $s$ describing the AoI at the receiver.

To solve Problem~\ref{problem_EH}, we need to find a policy for the transmitter, determining its actions for every state $s \in \mathcal{S}$, which can minimize the average AoI at the receiver.





It is easy to see that MDP formulated for Problem~\ref{problem_EH} is a communicating MDP by Proposition 8.3.1 of~\cite{Puterman_book}\footnote{By Proposition 8.3.1 of~\cite{Puterman_book}, MDP is communicating since there exists a stationary policy which induces a recurrent Markov chain, e.g., a state $(0,B_0,\Delta_{max},\Delta_{max},R_0)$ is reachable from all other states considering a policy which never transmits and in a scenario where no energy is harvested.}, i.e., for every pair of $(s,s')\in\mathcal{S}$, there exists a deterministic policy under which $s'$ is accessible from $s$. By Theorem 8.3.2 of~\cite{Puterman_book}, an optimal stationary policy exists with constant gain. In particular, there exists a function $h: \mathcal{S} \to \mathbb{R}$, called the \textit{differential cost function} for all $s=(e,b,\delta^{rx},\delta^{tx},r) \in \mathcal{S}$, satisfying the following \emph{Bellman optimality equations} for the average-cost finite-state finite-action MDP \cite{Puterman_book}:
\begin{align}
\label{eq:Bellman_EH}
h(s)+J^{*}&=\min_{a\in\{\idle,\sense,\retx\}}\big(\delta^{rx}+\Exp{h(s')|a}\big),
\end{align}
where $s'\triangleq (e',b',{\delta^{rx}}',{\delta^{tx}}',r')$ is the next state obtained from $(e,b,\delta^{rx},\delta^{tx},r)$ after taking action $a$, and $J^*$ represents the optimal achievable average AoI under policy $\pi^*$. Note that the function $h$ satisfying \eqref{eq:Bellman_EH} is unique up to an additive factor, and with selecting this additive factor properly, it also satisfies
\begin{align*}
 h(s)  = \Exp{\sum_{t=0}^\infty (\Delta^{rx}_t - J^{*})\big| S_0=s}. \end{align*}
We also introduce the \textit{state-action cost function}: 
\begin{align}
\lefteqn{Q((e,b,\delta^{rx},\delta^{tx},r),a)} \nonumber \\
& \triangleq \delta^{rx}+\Exp{h(e',b',{\delta^{rx}}',{\delta^{tx}}',r')|a}.
\label{eq:Bellman2_EH}
\end{align}
Then an optimal policy, for any $(e,b,\delta^{rx},\delta^{tx},r) \in \mathcal{S}$, takes the action achieving the minimum in \eqref{eq:Bellman2_EH}:
\begin{align}
\label{eq:opt_eta_EH}
\pi^*(e,b,\delta^{rx},\delta^{tx},r) &\in \argmin_{a\in\{\idle,\sense,\retx\}} \big(Q((e,b,\delta^{rx}\!,\!\delta^{tx}\!,\!r),a)\big). 
\end{align}

An optimal policy solving \eqref{eq:Bellman_EH}, \eqref{eq:Bellman2_EH} and \eqref{eq:opt_eta_EH} defined above can be found by relative value iteration (RVI) for finite-state finite-action average-cost MDPs from Sections 8.5.5 and 9.5.3 of \cite{Puterman_book}: Starting with an arbitrary initialization of $h_0(s)$, $\forall s \in \mathcal{S}$, and setting an arbitrary but fixed reference state $s^{ref}\triangleq (e^{ref},b^{ref},{\delta^{rx}}^{ref},{\delta^{tx}}^{ref},r^{ref})$,  a single iteration of the RVI algorithm $\forall (s,a) \in \mathcal{S}\times \mathcal{A}$ is given as follows:     \begin{align}
Q_{n+1}(s,a) &\leftarrow \Delta^{rx}_n+	\Exppi{h_n(s')},\\
V_{n+1}(s) &\leftarrow \min_{a}(Q_{n+1}(s,a)),\\
h_{n+1}(s) &\leftarrow {V}_{n+1}(s) -{V}_{n+1}(s^{ref}),
\end{align} 
where $Q_n(s,a)$, $V_n(s)$ and $h_n(s)$ denote the state-action value function, value function and differential value function at iteration $n$, respectively. By Theorem 8.5.7 and Section 8.5.5 of \cite{Puterman_book},  $h_n$ converges to $h$, and $\pi_{n}^*(s) \triangleq \argmin_{a} Q_{n}(s,a)$ converges to $\pi^*(s)$.

\subsection{Structure of the Optimal Policy}
\label{sec:structure_EH}

Next, we present the structure of the optimal policy and show that the solution to the Problem~\ref{problem_EH} is of threshold-type.
\begin{theorem}
There exits an optimal stationary policy  $\pi^*(s)$ that is monotone with respect to $\Delta_t^{rx}$, that is, $\pi^*(s)$ is of threshold-type.
\label{thm_monotonicty}
\end{theorem}
\begin{proof}
The proof is given in  Appendix~\ref{AppendixK}.
\end{proof}

Following Theorem~\ref{thm_monotonicty}, we present a threshold-based policy which will be termed as a \textit{double-threshold policy} in the rest of this paper.  
\begin{align}
A_t= \begin{cases}
\idle &\textrm{ if }  \Delta^{rx}_t < \mathcal{T_{\new}}(e,b,\delta^{tx},r), \\
\new  &\textrm{ if } \mathcal{T_{\new}}(e,b,\delta^{tx},r) \leq \Delta^{rx}_t < \mathcal{T_{\retx}}(e,b,\delta^{tx},r), \\
\retx  &\textrm{ if }  \Delta^{rx}_t \geq \mathcal{T_{\retx}}(e,b,\delta^{tx},r),
\end{cases} \label{eq:double_threshold}
\end{align} 
for some threshold values denoted by  $\mathcal{T_{\new}}(e,b,\delta^{tx},r)$ and $\mathcal{T_{\retx}}(e,b,\delta^{tx},r)$, where $\mathcal{T_{\new}}(e,b,\delta^{tx},r)\leq  \mathcal{T_{\retx}}(e,b,\delta^{tx},r)$. 

We can simplify the problem by making an assumption on the policy space in order to obtain a simpler \textit{single-threshold policy}, which will result in a more efficient learning algorithm: We assume that a packet is retransmitted until it is successfully decoded, provided that there is enough energy in the battery, that is, the transmitter is not allowed to preempt an undecoded packet and transmit a new one. 


The solution to the simplified problem is also of threshold-type,  that is, 
\begin{align}
A_t= \begin{cases}
\idle &\textrm{ if }  \Delta^{rx}_t < \mathcal{T}(e,b,\delta^{tx},r), \\
\new  &\textrm{ if }  \Delta^{rx}_t \geq \mathcal{T}(e,b,\delta^{tx},r), \textrm{ and } r=0 \\
\retx  &\textrm{ if }  \Delta^{rx}_t \geq \mathcal{T}(e,b,\delta^{tx},r) \textrm{ and } r\neq 0,
\end{cases}
\end{align} 
for some $\mathcal{T}(e,b,\delta^{tx},r)$.  

In Section \ref{sec:policy_gradient}, we present a RL algorithm to find the threshold values defined in this section.

\section{RL Approach}
\label{sec:learning_EH}

In some scenarios, it can be assumed that the channel and energy arrival statistics remain the same or change very slowly and the same  environment is experienced for a sufficiently long time before the time of deployment. Accordingly,  we can assume that the statistics regarding the error probabilities and energy arrivals are available before the time of transmission. In such scenarios, RVI algorithm presented in Section~\ref{sec:dynamic} can be used. However, in most practical scenarios, channel error probabilities for retransmissions and the EH characteristics are not known at the time of deployment, or may change over time. In this section, we assume that the transmitter  does not know the system characteristics \textit{a-priori}, and has to learn them. In our previous works \cite{wcnc_paper,journal_paper,pimrc_paper}, we have employed learning algorithms for constrained problems with countably infinite state spaces such as average-cost SARSA. While these algorithms can be adopted to the current framework by considering an average transmission constraint of 1, they do not have convergence guarantees. However, Problem~\ref{problem_EH} constitutes an unconstrained MDP with finite state and action spaces, and there exist RL algorithms for unconstrained MDPs which enjoy convergence guarantees. Moreover, we have shown the optimality of a threshold type policy for Problem~\ref{problem_EH}, and RL algorithms which exploit this structure can be considered.  Thus, we employ three different RL algorithms, and compare their performances in terms of the average AoI as well as the convergence speed. First, we employ a value-based RL algorithm, namely GR-learning, which converges to an optimal policy. Next,  we consider a structured policy search algorithm, namely FDPG, which does not necessarily find the optimal policy but performs very well in practice, as demonstrated through simulations in Section~\ref{sec:results_EH}. We also note that GR-learning learns from a single trajectory generated during learning steps while FDPG uses Monte-Carlo roll-outs for each policy update. Thus, GR-learning is more amendable for applications in real-time systems. Finally, we employ the DQN algorithm, which implements a non-linear neural network estimator in order to learn the optimal status update policy. 
 

\subsection{GR-Learning with Softmax}
\label{sec:gr}

The literature for average-cost RL is quite limited compared to discounted cost problems  \cite{Sutton1998,Mahadevan1996}.  For the average AoI minimization problem in \eqref{eq:problem}, we employ a modified version of the  \textit{GR-learning} algorithm proposed in \cite{Gosavi2004}, as outlined in Algorithm~\ref{algo_learn}, with \emph{Boltzmann}  (\emph{softmax}) exploration.  The resulting algorithm is called \emph{GR-learning with softmax}.

Notice that, by only knowing $Q(s,a)$, one can find the optimal policy $\pi^*$ using \eqref{eq:opt_eta_EH} without knowing the transition probabilities $P$, characterized by $g(r)$ and $p_E$. Thus, \emph{GR-learning with softmax} starts with an initial estimate of $Q_{0}(s,a)$ and finds the optimal policy by estimating state-action values in a recursive manner. In the $n^{th}$ iteration, after taking action $A_n$, the transmitter observes the next state $S_{n+1}$ and the instantaneous cost value $\Delta^{rx}_n$. Based on these, the estimate of $Q_{n+1}(s,a)$ is updated by a weighted average of the previous estimate $Q_{n}(s,a)$ and the estimated expected value of the current policy in the next state $S_{n+1}$. Moreover, we update the gain $J_n$ at every time slot based on the empirical average of AoI. 

In each time slot, the learning algorithm 
\begin{itemize}
\item observes the current state $S_n \in \mathcal{S}$,
\item selects and performs an action $A_n \in \mathcal{A}$,
\item observes the next state $S_{n+1}\in \mathcal{S}$ and the instantaneous cost $\Delta^{rx}_n$,
\item updates its estimate of $Q(S_n,A_n)$ using the current estimate of $J_{n}$ by
\begin{align}
Q_{n+1}(S_n,A_n)\leftarrow Q_{n}(S_n,A_n) +  \alpha(m(S_n,A_n,n)) \nonumber \\ [\Delta^{rx}_n -J_{n}+Q_{n}(S_{n+1},A_{n+1})-Q_{n}(S_n,A_n)],
\end{align}
where $\alpha(m(S_n,A_n,n))$ is the update parameter (learning rate) in the $n^{th}$ iteration, and depends on the function $m(S_n,A_n,n)$, which is the number of times the state–action pair $(S_n, A_n)$ has been visited till the $n^{th}$ iteration.
\item updates its estimate of $J_{n}$ based on the empirical average as follows:
\begin{align}
  J_{n+1}\leftarrow J_n+ \beta(n) \left[\frac{n J_n+\Delta^{rx}_n}{n+1}-J_n\right]  \end{align}
 where $\beta(n)$ is the update parameter in the $n^{th}$ iteration. 
  \end{itemize}

The transmitter action selection method should balance the \textit{exploration} of new actions with the \textit{exploitation} of actions known to perform well. We use the \textit{Boltzmann} (\textit{softmax}) action selection method, which chooses each action randomly relative to its expected cost as follows:
\begin{equation}
    \pi(a|S_n)=\frac{\displaystyle\exp(-Q(S_n,a)/\tau_n)}{\displaystyle\sum_{a'\in\mathcal{A}}{\exp(-Q(S_n,a')/\tau_n)}} \label{eq:softmax}.
\end{equation}
Parameter $\tau_n$ in \eqref{eq:softmax} is called the temperature parameter and decays exponentially with decay parameter $\gamma_{\tau}\leq 1$ at each iteration. High $\tau_n$ corresponds to more uniform action selection (exploration) whereas low $\tau_n$ is biased toward the best action (exploitation). According to Theorem~2 of \cite{Gosavi2004}, if $\alpha$, $\beta$ satisfy $\sum_{m=1}^{\infty}\alpha(m), \sum_{m=1}^{\infty}\beta(m)\rightarrow \infty$,   $\sum_{m=1}^{\infty}\alpha^2(m), \sum_{m=1}^{\infty}\beta^2(m) <\infty$, $
\lim_{x \to \infty} \frac{\beta(m)}{\alpha(m)}\rightarrow 0$, \textit{GR}-Learning converges to an optimal policy. 

\begin{algorithm}
\caption{GR-learning with softmax}
\begin{small}
\begin{algorithmic}[1]
 \renewcommand{\algorithmicrequire}{\textbf{Input:}}
 \renewcommand{\algorithmicensure}{\textbf{Output:}}
 \REQUIRE error probabilities $g(r)$ and harvesting probabilities $p_E$ are unknown
    \STATE $\tau_0\leftarrow 1$,     \Comment{temperature parameter}  
    \STATE $\gamma \leftarrow 0.95$,   \Comment{softmax decay coefficient}.
    \STATE $Q_0 \leftarrow 0$, $\forall_{(s,a)}$   \Comment{initialization of $Q$} 
    \STATE $J_0 \leftarrow 0$,   \Comment{initialization of the gain} 
\FOR{$n=\{1,2,\ldots\}$}   
    \STATE \textsc{Observe} the current state $S_n$

    \FOR{$a \in \mathcal{A}$}   
   
   $\pi(a|S_n)=\frac{\displaystyle\exp(-Q(S_n,a)/\tau_n)}{\displaystyle\sum_{a'\in\mathcal{A}}{\exp(-Q(S_n,a')/\tau_n)}}$ 
   \ENDFOR
   
    \STATE \textsc{Sample} and \textsc{Perform} $A_n$ from $\pi(a|S_n)$ 
    
    \STATE \textsc{observe} the next state $S_{n+1}$ and cost $\Delta_t$ 
    
    \STATE \textsc{Update}  the estimates of $Q(S_n,A_n)$ and $J_{n}$ by
    \begin{align*}
    Q(S_n,A_n)&\leftarrow Q(S_n,A_n) + \alpha(m(S_n,A_n,n)) [\Delta_t-J_n\\&~~+\min_{A_{n+1}}{Q(S_{n+1},A_{n+1}})-Q(S_n,A_n)]
   \\
    J_{n+1}&\leftarrow J_n+ \beta(n) [\frac{n J_n+\Delta_t}{n+1}-J_n]  
    \end{align*} 
    
    \STATE \textsc{Update} step size parameters
    \begin{align*}
     \tau_{n+1} &\leftarrow \gamma \tau_n \\
     \alpha(n)&\leftarrow 1/\sqrt{m(S_n,A_n,n)}   \\
    \beta(n)&\leftarrow 1/n \\
    m(S_n,A_n,n+1) &\leftarrow  m(S_n,A_n,n)+1\\
    m(s,a,n+1) &\leftarrow  m(s,a,n), \forall (s,a)\neq (S_n,A_n)
    \end{align*}
    
 \ENDFOR

\end{algorithmic}
\end{small}
\label{algo_learn}
\end{algorithm}




\subsection{Finite-Difference Policy Gradient (FDPG)}
\label{sec:policy_gradient}

GR-learning in Section~\ref{sec:gr} is a value-based RL method, which learns the state-action value function for each state-action pair. In practice, $\Delta_{max}$ can be large, which might slow down the convergence of GR-learning due to a prohibitively large state space. 


In this section, we are going to propose a learning algorithm which exploits the structure of the optimal policy exposed in Theorem~\ref{thm_monotonicty}. A monotone policy is shown to be average optimal in the previous section; however, it is not possible to compute the threshold values directly for Problem~\ref{problem_EH}. 


Note that, $A_t=\idle$ if $B_t<E^{\mathrm{tx}}$  ($B_t<E^{\mathrm{tx}}+E^{\mathrm{s}}$) for $r\geq 1$ ($r=0$); that is,  $\mathcal{T}(e,b,\delta^{tx},r)=\Delta_{max}+1$ if $b<E^{\mathrm{tx}}$ for $r\geq 1$ and $b<E^{\mathrm{tx}}+E^{\mathrm{s}}$ for $r=0$. This ensures that energy causality constraints in \eqref{eq:causality2} hold. Other thresholds will be determined using policy gradient. 

In order to employ the policy gradient method, we approximate the policy by a parameterized smooth function with parameters $\theta(e,b,\delta^{tx},r)$, and convert the discrete policy search problem into estimating the optimal values of these continuous parameters, which can be numerically solved by  stochastic approximation algorithms \cite{Spall2003}.  Continuous stochastic approximation is much more efficient than discrete search algorithms in general.

In particular, with a slight abuse of notation, we let $\pi_{\theta}(e,b,\delta^{rx},\delta^{tx},r)$ denote the probability of taking action $A_t=\sense$ ($A_t=\retx$) if $r=0$ ($r\neq 0$), and  consider the parameterized sigmoid function:
\begin{align}
\pi_{\theta}(e,b,\delta^{rx},\delta^{tx},r)\triangleq \frac{1}{1+e^{-\frac{\delta^{rx}-\theta(e,b,\delta^{tx},r)}{\tau}}}.
\end{align}
We note that $\pi_{\theta}(e,b,\delta^{rx},\delta^{tx},r) \rightarrow \{0,1\}$ and $\theta(e,b,\delta^{tx},r)\rightarrow \mathcal{T}(e,b,\delta^{tx},r)$ as $\tau \rightarrow 0$. Therefore, in order to converge to a deterministic policy $\pi$, $\tau>0$ can be taken as a sufficiently small constant, or can be decreased gradually to zero.  The total number of parameters to be estimated is $|\mathcal{E}| \times B_{\mathrm{max}}\times \Delta_{max}\times R_{max}+1$ minus the parameters corresponding to  $b<E^{\mathrm{tx}}$  ($b<E^{\mathrm{tx}}+E^{\mathrm{s}}$) for $r>0$ ($r=0$) due to energy causality constraints as stated previously. 

With a slight abuse of notation, we map the parameters $\theta(e,b,\delta^{tx},r)$ to a vector $\overline{\theta}$ of size $d\triangleq |\mathcal{E}| \times B_{\mathrm{max}} \times \Delta_{max}\times R_{max}+1$.  Starting with some initial estimates of $\overline{\theta}_0$, the parameters can be updated in each iteration $n$ using the gradients as follows: 
\begin{align}
   \overline{\theta}_{n+1}=\overline{\theta}_n-\gamma(n) ~{\partial J}/{\partial \overline{\theta}_n},
\end{align}
where the step size parameter $\gamma(n)$ is a positive decreasing sequence and satisfies the first two convergence properties given at the end of Section~\ref{sec:gr} from the theory of stochastic approximation \cite{stochastic_approx}.

Computing the gradient of the average AoI directly is not possible; however, several methods exist in the literature to estimate the gradient \cite{Spall2003}. In particular, we employ the FDPG  \cite{policy_gradient} method. In this method, the gradient is estimated by estimating $J$ at slightly perturbed parameter values. First, a random perturbation vector $D_n$ of size $d$ is generated according to a predefined probability distribution, e.g., each component of $D_n$ is an independent Bernoulli random variable with parameter $q\in (0,1)$. The thresholds are perturbed with a small amount $\sigma>0$ in the directions defined by $D_n$ to obtain $\overline{\theta}_n^{\pm}(e,b,\delta^{tx},r)\triangleq \overline{\theta}_n(e,b,\delta^{tx},r)\pm \sigma D_n$. Then, empirical estimates $\hat{J}^{\pm}$ of the average AoI corresponding to the perturbed parameters $\overline{\theta}_n^{\pm}$, obtained from Monte-Carlo rollouts, are used to estimate the gradient:
\begin{align}
  {\partial J}/{\partial \overline{\theta}_n} \approx (D_n^{\intercal} D_n)^{-1} D_n^{\intercal} \frac{(\hat{J}^+-\hat{J}^-)}{2 \sigma}, 
\end{align}
where $D_n^{\intercal}$ denotes the transpose of vector $D_n$. The pseudo code of the finite difference policy gradient algorithm is given in Algorithm~\ref{algo_learn2}.

\begin{algorithm}
\caption{FDPG}
\begin{small}
\begin{algorithmic}[1]
 \renewcommand{\algorithmicrequire}{\textbf{Input:}}
 \renewcommand{\algorithmicensure}{\textbf{Output:}}
 \REQUIRE error probabilities $g(r)$ and harvesting probabilities $p_E$ are unknown
    \STATE $\tau_0\leftarrow 0.3$,     \Comment{temperature parameter}  
    \STATE $\zeta\leftarrow 0.99$,   \Comment{decaying coefficient for $\tau$}.
    \STATE $\overline{\theta}_0 \leftarrow 0$   \Comment{initialization of $\overline{\theta}$} 
\FOR{$n=\{1,2,\ldots\}$}   
    \STATE \textsc{Generate} random perturbation vector
    
    $D_n = binomial(\{0,1\}, q=0.5, d)$ 
    
        \STATE \textsc{Perturb} parameters $\overline{\theta}_n$
    
    $\overline{\theta}_n^+=\overline{\theta}_n+\beta D_n$,
    $\overline{\theta}_n^-=\overline{\theta}_n-\beta D_n$
    
    
        \STATE \textsc{Estimate} $\hat{J}_n^{\pm}$ from Monte-Carlo rollouts using policies $\pi_{\theta_n^{\pm}}$:
    

    \FOR{$t \in \{1,\ldots, T\}$ }  
        \STATE \textsc{observe} current state $S_t$ and
         \textsc{use}  policy $\pi_{\theta_n^{\pm}}$
     \ENDFOR
    
        \STATE \textsc{Estimate} $\hat{J}_n^{\pm}$ from Monte-Carlo rollouts using policy $\pi_{\theta_n^{\pm}}$
    
    $\hat{J}_n^{\pm}=\frac{1}{T}\sum_{t=1}^T{\Delta^{rx}_t}$
    
        \STATE \textsc{Compute} the estimate of the gradient ${\partial J}/{\partial \overline{\theta}_n}$
    
    
        \STATE \textsc{Update} 
    
    $\overline{\theta}_{n+1}=\overline{\theta}_n-\gamma(n) ~{\partial J}/{\partial \overline{\theta}_n}$
    
    $\tau_{n+1}\leftarrow  \zeta  \tau_n$ ~\Comment{decrease $\tau$}
    
    
 \ENDFOR

\end{algorithmic}
\end{small}
\label{algo_learn2}
\end{algorithm}


Similar steps can be followed to find the thresholds for the double-threshold policy where $\mathcal{T}(e,b,\delta^{tx},r)$ and $\theta(e,b,\delta^{tx},r)$ are replaced by $\mathcal{T_{\new}}(e,b,\delta^{tx},r)$, $\mathcal{T_{\retx}}(e,b,\delta^{tx},r)$ and $\theta_{\new}(e,b,\delta^{tx},r)$, $\theta_{\retx}(e,b,\delta^{tx},r)$ respectively.

\subsection{Deep Q-Network (DQN)}

A DQN uses a multi-layered neural network in order to estimate the values $Q(s,a)$ of the underlying MDP; that is, for a given state $s$, DQN outputs a vector of state-action values, $Q_{\theta}(s, a)$, where $\theta$ denotes the parameters of the network. The neural network is a function from $2M$ inputs to $|\mathcal{A}|$ outputs, which are the estimates of the Q-function $Q_{\theta}(s,a)$.  We apply the DQN algorithm of \cite{Mnih2015} to learn a scheduling policy. We create a simple feed-forward neural network of $3$ layers, one of which is the hidden layer with 24 neurons. We use \textit{Huber loss} \cite{huber1964} and the \textit{Adam} algorithm \cite{Adam2014} to conduct stochastic gradient descent to update the weights of the neural network. 

We exploit two important features of DQNs  as proposed in \cite{Mnih2015}: \textit{experience replay} and a \textit{fixed target network}, both of which provide algorithm stability. For \textit{experience replay}, instead of training the neural network with a single observation $<\!s,a,s',c(s,a)\!>$ at the end of each step,  many experiences (i.e., (state, action, next state, cost) quadruplets) can be stored in the replay memory for batch training, and a minibatch of observations randomly sampled at each step can be used. The DQN uses two neural networks: a target network and an online network. The \textit{target network}, with parameters $\theta^-$, is the same as the online network except that its parameters are updated with the parameters $\theta$ of the online network after every $T$ steps, and $\theta^-$ is kept fixed in other time slots. For a minibatch of observations for training, temporal difference estimation error $e$ for a single observation can be calculated as 
\vspace{-0.1in}
\begin{align}
    \varepsilon \!=\! Q_{\theta}(s\!,\!a)\!-\!(\!-\!c(s,a)+\gamma Q_{\theta^-}(s'\!,\!\argmax Q_{\theta}(s'\!,\!a))).
\end{align} 

\textit{Huber loss} is defined by the squared error term for small estimation errors, and a linear error term for high estimation errors, allowing less dramatic changes in the value functions, further improving the stability. For a given estimation error $\varepsilon$ and loss parameter $d$, the Huber loss function, denoted by $\mathrm{L}^d(\varepsilon)$ is defined as:
\begin{align*}
\mathrm{L}^d(\varepsilon)=
    \begin{cases}
    \varepsilon^2 &\textrm { if } \varepsilon\leq d \\
    d (|\varepsilon|-\frac{1}{2}d)) &\textrm { if } \varepsilon> d,
    \end{cases}
\end{align*}
and loss over minibatch $\mathcal{B}$ is computed as:
\begin{align*}
    \mathrm{L}_{\mathcal{B}}= \frac{1}{|\mathcal{B}|}\sum_{<s,a,s',c(s,a)>\in \mathcal{B}} \mathrm{L}^d(\varepsilon).
\end{align*}

We apply the $\epsilon$-greedy policy to balance
 exploration and exploitation, i.e., with probability $\epsilon$ the source randomly selects an action, and with probability $1-\epsilon$ it chooses the action with the minimum Q value.  We let $\epsilon$ decay gradually from $\epsilon_0$ to $\epsilon_{min}$; in other words, the source explores more at the beginning of training and exploits more at the end. The hyperparameters of the DQN algorithm are  tuned for our problem experimentally, and are presented in Table~\ref{table_DQN}.

\begin{table}
\caption{Hyperparameters of DQN algorithm used in the paper}
\begin{footnotesize}
\begin{center}
 \label{table_DQN}
 \scriptsize
\begin{tabular}{ |c|c||c|c| } 
\hline
 Parameter & Value & Parameter & Value \\ 
 \hline \hline
 discount factor $\gamma$ & 0.99 & optimizer & Adam \\ 
 \hline
 minibatch size & 32 & loss function & Huber loss\\ 
 \hline
 replay memory length & 2000 & exploration coeff. $\epsilon_0$ & 1\\ 
 \hline
 learning rate $\alpha$  & $10^{-4}$ & $\epsilon$ decay rate $\beta$ & 0.9\\ 
 \hline
  episode length $T$ & 1000 & $\epsilon_{min}$ & 0.01\\ 
 \hline
 activation function & ReLU & hidden size & 24\\ 
 \hline
\end{tabular}
\end{center}
\end{footnotesize}
\end{table}

\section{Simulation Results}
\label{sec:results_EH}

In this section, we provide numerical results for all the proposed algorithms, and compare the achieved average AoI. Motivated by previous research on HARQ \cite{hybrid2001}, \cite{harq2003}, \cite{IEEEstandard}, we assume that the decoding error reduces exponentially with the number of retransmissions, that is, $g(r)\triangleq p_0 \lambda^{r}$ for some $\lambda \in (0,1)$, where $p_0$ denotes the error probability of the first transmission and $r$ is the retransmission count (set to $0$ for the first transmission). 
The exact value of the rate $\lambda$ depends on the particular HARQ protocol and the channel model. Following the \emph{IEEE 802.16} standard\cite{IEEEstandard}, the maximum number of retransmissions used for decoding is set to $R_{max}=3$. In the following experiments, $\lambda$ and $p_0$ are set to $0.5$. $E^{\mathrm{tx}}$ and $E^{s}$ are both assumed to be constant and equal to 1 unit of energy unless otherwise stated. $\Delta_{max}$ is set to $40$. 

We choose the exact step sizes for the learning algorithms by fine-tuning in order to balance the algorithm stability in the early time steps with nonnegligible step sizes in the later time steps.  In particular, we use step size parameters of  $\alpha(m),\beta(m),\gamma(m)=y/(m+1)^z$,  where $0.5<z\leq 1$ and $y>0$ (which satisfy the convergence conditions), and choose $y$ and $z$ such that the oscillations are low and the convergence rate is high. We have observed that a particular choice of parameters gives similar performance results for scenarios addressed in simulations results. 

DQN algorithm in this section is configured as in Table~\ref{table_DQN} and trained for $500$ episodes. The average AoI for DQN is obtained after $10^5$ time steps and averaged over $100$ runs.

As a baseline, we have also included the performance of a greedy policy, which senses and transmits a new status update whenever there is sufficient energy. It retransmits the last transmitted status update when the energy in the battery is sufficient
only for transmission, and it remains idle otherwise; that is, 
\begin{align}
A^{greedy}_t=
    \begin{cases}
     \idle &\textrm{if } B_t<E^{\mathrm{tx}}, \\
    \sense &\textrm{if } B_t\geq E^{\mathrm{tx}}+E^{s}, \\
   \retx &\textrm{if } E^{\mathrm{tx}}\leq B_t< E^{\mathrm{tx}}+E^{s}. \\
    \end{cases}
\end{align}


\subsection{Memoryless EH Process}

We first investigate the average AoI with HARQ when the EH process, $E_t\in \mathcal{E}=\{0,1\}$, is i.i.d. over time with probability distribution  $Pr(E_t=1)=p_e$, $\forall t$. Figure~\ref{fig:policy_iid} illustrates the policy obtained by the RVI algorithm in Section~\ref{sec:dynamic}. The resulting policy is more likely to transmit if the battery level or the AoI is high as expected. Moreover, the policy tends to retransmit the previous update rather than sensing a new update when the battery level is low and the AoI is high. We can also observe from the figure that the optimal policy exhibits a threshold structure as shown in Theorem~\ref{thm_monotonicty}.
\begin{figure}
    \centering
    \includegraphics[scale=0.4]{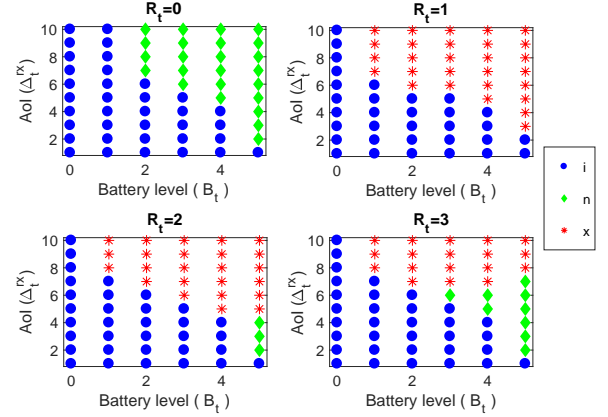}
    \caption{Optimal policy for memoryless EH when $B_{\mathrm{max}}=5$, $R_{max}=3$, $p_e=0.5$ and $E^{\mathrm{s}}=E^{\mathrm{tx}}=1$. The decoding error probabilities are given by $g(r)= 2^{-(r+1)}$.}
    \label{fig:policy_iid}
\end{figure}

 The effects of the battery capacity $B_{\mathrm{max}}$, energy consumption of sensing $E^{\mathrm{s}}$, and the energy harvesting probability $p_e$ on the average AoI are shown in Figure~\ref{fig:effectofB}. As expected, the average AoI increases with decreasing $B_{\mathrm{max}}$, decreasing $p_e$ and increasing $E^{\mathrm{s}}$.  We note that, when $E^{\mathrm{s}}=0$ and $B_{\mathrm{max}}=\infty$, the problem defined in~\eqref{eq:problem} corresponds to minimizing the average AoI under an average transmission rate constraint $p_e$, studied in \cite{wcnc_paper,journal_paper}. The average AoI under average transmission rate constraint ($B_{\mathrm{max}}=\infty$) is also shown in Figure~\ref{fig:effectofB} and we also observe that its performance can be approximated with a finite battery of size of $B_{\mathrm{max}}=30$ at low $p_e$ values, while a battery size of $B_{\mathrm{max}}=5$ is sufficient when $p_e$ increases.

\begin{figure}
\centering
\includegraphics[scale=0.4]{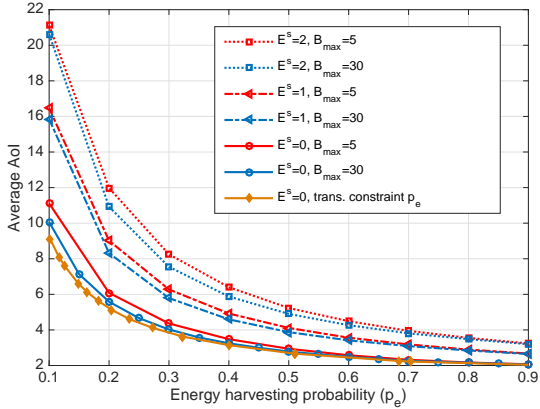}
\caption{Average AoI for different $B_{\mathrm{max}}$, $E^{\mathrm{s}}$ and $p_e$ values for memoryless EH and $E^{\mathrm{tx}}=1$.}
\label{fig:effectofB}
\end{figure}

Figure~\ref{fig:iid05} shows the evolution of the average AoI over time when the average-cost RL algorithms are employed. It can  be observed that the average AoI achieved by the proposed RL algorithms, converge to values close to the one obtained from the RVI algorithm, which has \emph{a priori} knowledge of $g(r)$ and $p_e$, while the AoI of the greedy algorithm is significantly higher. Although GR-learning enjoys theoretical guaranties to converge to the optimal policy, the FDPG which benefits from the structural guarantees of a threshold policy (including a single-threshold policy not allowing preemption of an undecoded status update), performs better than GR-learning since it tries to learn significantly smaller number of threshold values (i.e., $\Delta_{max}\times  B_{\mathrm{max}}\times R_{max}+1$) compared to GR-learning, which learns one value for each state-action pair (i.e., $\Delta_{max}^2 \times B_{\mathrm{max}}\times (R_{max}+1) \times |\mathcal{A}| $).  We also observe that, among the FDPG methods, the one with a single threshold converges faster but the double-threshold policy finally attain on slightly lower AoI. Therefore, the choice between the two may depend on the stochasticity of  the environment. DQN algorithm performs better than GR-learning but it requires a training time before running the simulation and does not have convergence guarantees. Moreover, its final performance is slightly worse than both of the FDPG algorithms.



\begin{figure}
    \centering
    \includegraphics[scale=0.4]{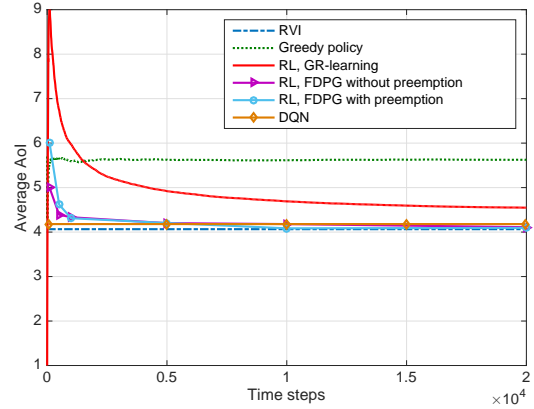}
    \caption{Performance of RL algorithms when $B_{\mathrm{max}}=5$, $E^{\mathrm{s}}=E^{\mathrm{tx}}=1$, and $p_e=0.5$. FDPG with and without preemption represent the double-threshold and the single-threshold policies, respectively.}
    \label{fig:iid05}
\end{figure}

\subsection{Temporally Correlated EH}

Next, we investigate the performance when the EH process has temporal correlations. A symmetric two-state Markovian EH process is assumed, such that $\mathcal{E}=\{0,1\}$  and $Pr(E_{t+1}=1|E_t=0)=Pr(E_{t+1}=0|E_t=1)=0.3$. That is, if the transmitter is in harvesting state, it is more likely to continue harvesting energy, and vice versa for the non-harvesting state. 

Figure~\ref{fig:policy} illustrates the policy obtained by RVI. As it can be seen from the figure, the resulting policy is less likely to transmit if the battery level or the AoI is low. As shown in Theorem~\ref{thm_monotonicty}, the optimal policy exhibits a threshold structure on $\Delta^{rx}$. Moreover, the policy tends to retransmit the previous update rather than sensing a new update when the battery level is low and the AoI is high. When the system is in the non-harvesting state (i.e., $E_t=0$), the transmitter is more conservative in transmitting the status updates compared to the case $E_t=1$, e.g., it might not transmit even if the battery is full depending on the AoI level.  

\begin{figure}
\centering
\begin{subfigure}
\centering
\includegraphics[scale=0.4]{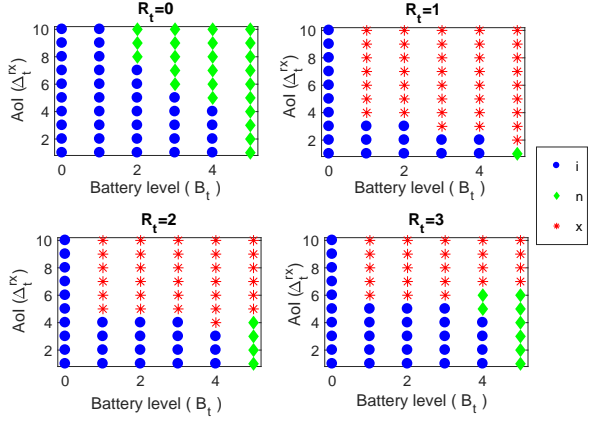}
\end{subfigure} \\
(a) $E_t=1$ \\
\begin{subfigure}
\centering
\includegraphics[scale=0.4]{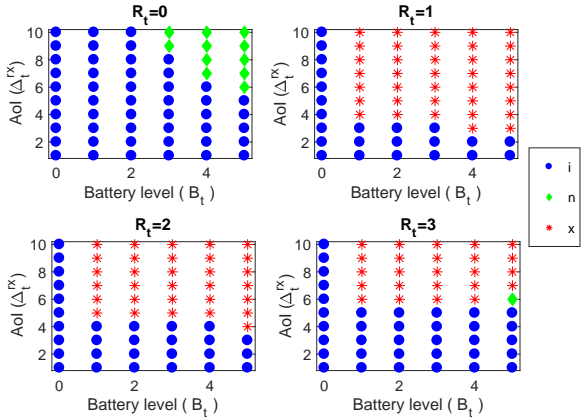}
\end{subfigure}\\
(a) $E_t=0$ \\
\caption{Optimal policy for $B_{\mathrm{max}}=5$, $R_{max}=3$, $p_E(1,1)$, $p_E(0,0)=0.7$, $E^{\mathrm{s}}=E^{\mathrm{tx}}=1$ and $\Delta^{tx}_t=R_t+1$. The decoding error probabilities are given by $g(r)= 2^{-(r+1)}$.}
\label{fig:policy}
\end{figure}

Figure~\ref{fig:learn_corr_time} shows the evolution of the average AoI over time when the average-cost RL algorithms are employed  in this scenario. It can  be observed again that the average AoI achieved by the FDPG method in Section~\ref{sec:policy_gradient} performs very close to the one obtained by the RVI algorithm, which has \emph{a priori} knowledge of $g(r)$ and $p_e$. GR-learning, on the other hand, outperforms the greedy policy but converges to the optimal policy much more slowly, and the gap between the two RL algorithms is  larger compared to the i.i.d. case. Tabular methods in RL, like GR-learning, need to visit each state-action pair infinitely often for RL to converge \cite{Sutton1998}.  GR-learning in the case of  temporally correlated EH does not perform as well as in the i.i.d. case since the state space becomes larger with the addition of the EH state. We also observe that the gap between the final performances of single- and double-threshold FDPG solutions is larger compared to the memoryless EH scenario, while the single threshold solutions till converges faster. DQN algorithm performs better than GR-learning but it requires a training time before running the simulation and does not have convergence guarantees.  Moreover, it still falls short of the final performance of double-threshold FDPG.

Figure~\ref{fig:arq} illustrates the effect of preemption and the performance improvement of double-threshold FDPG over single-threshold FDPG for a scenario where preemption is 	inherently need, e.g., $g(r)$ is same for all retransmissions $r$ representing a standard ARQ protocol and dropping a failed update improves the performance. As it can be seen from Figure~\ref{fig:arq}, although single-threshold FDPG converges very close to the RVI without preemption, the average AoI obtained by single-threshold FDPG is still considerable higher than that of double-threshold FDPG for standard ARQ protocol. 
\begin{figure}[h]
    \centering
    \includegraphics[scale=0.4]{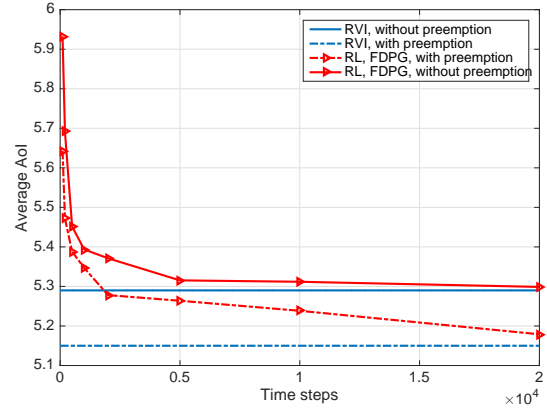}
    \caption{The performance of FDPG algorithms when $B_{\mathrm{max}}=5$, $R_{max}=3$, $p_E(1,1)$, $g(r)=0.5$, $\forall r$,  $p_E(0,0)=0.7$ and $E^{\mathrm{s}}=E^{\mathrm{tx}}=1$. FDPG with and without preemption represent the double-threshold and the single-threshold policies, respectively.}
    \label{fig:arq}
\end{figure}

\begin{figure}
    \centering
    \includegraphics[scale=0.4]{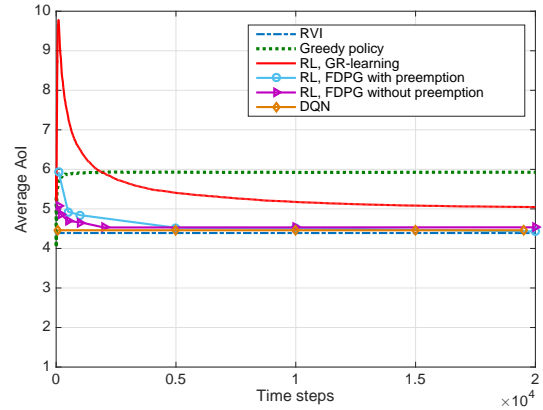}
    \caption{The performance of RL algorithms when $B_{\mathrm{max}}=5$, $p_E(1,1)$, $p_E(0,0)=0.7$ and $E^{\mathrm{s}},E^{\mathrm{tx}}=1$. FDPG with and without preemption represent the double-threshold and the single-threshold policies, respectively.}
    \label{fig:learn_corr_time}
\end{figure}
Next, we investigate the impact of the burstiness of the EH process, measured
by the correlation coefficient between $E_t$ and
$E_{t+1}$. Figure~\ref{fig:learn_corr} illustrates the performance of the proposed RL algorithms for different correlation coefficients, which can be computed easily for the 2-state symmetric Markov chain; that is, $\rho\triangleq (2p_E(1,1)-1)$. Note that $\rho=0$ corresponds to memoryless EH with $p_e=1/2$.   We note that the average AoI is minimized by transmitting new packets successfully at regular intervals, which has been well investigated in previous works \cite{Tan2015,wcnc_paper,Yates2015}. Intuitively, for highly correlated EH, there are either successive transmissions or successive idle time slots, which increases the average AoI. Hence, the AoI is higher for higher values of $\rho$.  Figure~\ref{fig:learn_corr} also shows that both RL algorithms result in much lower average AoI than the greedy policy and FDPG outperforms GR-learning since it benefits from the structural characteristics of a threshold policy.  We can also conclude that the single threshold policy can be preferable in practice especially in highly dynamic environments, as its performance is very close to that of the double threshold FDPG, but with faster convergence.
\begin{figure}
    \centering
    \includegraphics[scale=0.4]{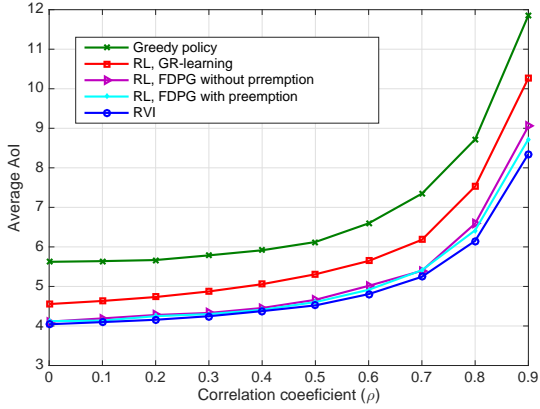}
    \caption{The performance of RL algorithms obtained after $2\cdot10^4$ time steps and averaged over $1000$ runs for different temporal correlation coefficients.}
    \label{fig:learn_corr}
\end{figure}

\section{Conclusions}
\label{sec:conclusion_EH}
We have considered an EH system with a finite size battery and investigated scheduling policies transmitting time-sensitive data over a noisy channel with the average AoI as the performance measure. We have assumed the presence of an Ack/NACK feedback from the receiver, and allowed retransmissions with an HARQ protocol to increase the probability of correct reception of status updates. This results in a trade-off between sending new status updates and retransmitting failed status updates as the former results in a lower AoI at the receiver while the latter is more likely to succeed. This trade-off is exacerbated in the model considered in this paper by the introduction of a sensing cost, which increases the cost of new status updates, and requires  judicious decisions at the transmitter due to limited and stochastic availability of energy.

In addition to identifying a RVI solution for the optimal policy when the system characteristics are known, efficient RL algorithms are  presented for practical applications when the system characteristics may not be known in advance. The effects of the battery size, EH characteristics and the HARQ structure on the average AoI are investigated through numerical simulations. 

We have presented three types of RL algorithms with different levels of complexity and training requirements and compared their performances for the current problem under a variety of system setting. We have observed that FDPG policies that exploit the threshold structure of the optimal policy provide both better performance and convergence behaviour. Moreover a simplified single threshold FDPG alternative is shown to further increase the convergence speed with a negligible increase in the average AoI.

\appendix

\subsection {Proof of Theorem~\ref{thm_monotonicty}:}

\label{AppendixK}

By \eqref{eq:Bellman2_EH} and \eqref{eq:opt_eta_EH}, Theorem~\ref{thm_monotonicty} holds if $Q((e,b,\delta^{rx},\delta^{tx},r),a)$ has a \textit{sub-modular} structure in $(\delta^{rx},a)$~\cite{Topkis1978}: that is, when the difference between the $Q$ function is monotone with respect to the state-action pair ($\delta^{rx},a$) for any $E^h_t=e$, $B_t=b$, $\Delta_t^{tx}=\delta^{tx}$, and $R_t=r$. We show  the submodularity by verifying the following inequality for 3 action pairs $(a_1,a_2)\in\{(\idle,\new),(\idle,\retx),(\new,\retx)\}$:
 \begin{align}
  Q(e,b,\delta^{rx}+1,\delta^{tx},r,a_2) -   Q(e,b,\delta^{rx}+1,\delta^{tx},r,a_1) \nonumber \\ \leq Q(e,b,\delta^{rx},\delta^{tx},r,a_2) -   Q(e,b,\delta^{rx},\delta^{tx},r,a_1)
  \label{eq:submodular}
  \end{align}
 
 Inequality~\eqref{eq:submodular} can be rewritten for $(a_1,a_2)=(\idle,\new)$ using \eqref{eq:Bellman2_EH},
 \begin{eqnarray}
   &Q(e,b,\delta^{rx},\delta^{tx},r,\sense)= \delta^{rx}+\sum_{e'\in \mathcal{E}}p_E(e,e') \nonumber \\& [g(r)h(e',b+e-E^s-E^{tx},\delta^{rx} +1,1,1)\nonumber \\&+(1-g(r))h(e',b+e-E^s-E^{tx},1,1,0) ]
   \label{eq:bel1}
  \end{eqnarray} and
  \begin{align}
   Q(e,b,\delta^{rx},\delta^{tx},r,\idle)=\delta^{rx} +\sum_{e'\in \mathcal{E}}p_E(e,e') \nonumber \\h(e',b+e,\delta^{rx} +1,\delta^{tx} +1,r)
   \label{eq:bel2}
  \end{align}
\eqref{eq:bel1} and \eqref{eq:bel2} is inserted into \eqref{eq:submodular} and since the next state $E^h_{t+1}= e'$ is independent of action ($A_t$) and  AoI ($\Delta_t$),  the following is equivalent to \eqref{eq:submodular}: 
 \begin{align}
  &g(0) \big[h(e',b+e-E^s-E^{tx},\delta^{rx}+2,1,1) \nonumber \\ &-  h(e',b+e-E^s-E^{tx},\delta^{rx}+1,1,1)\big] \nonumber \\
  &- \big[h(e',b+e,\delta^{rx}+2,\delta^{tx}+1,r) \nonumber \\&- h(e',b+e,\delta^{rx}+1,\delta^{tx}+1,r)\big] \leq 0.
  \label{eq:submodular2}
  \end{align}
  Also we note that $\Delta_{t+1}^{rx}=\delta^{rx}+2$, $\Delta_{t+1}^{tx}=\delta^{tx}+1$ and $R_{t}=r$ are truncated to $\Delta_{max}$, $\Delta_{max}$ and $R_{max}$ respectively.
  
  The same steps can be repeated for $(a_1,a_2)\in(\idle,\retx)$ and $(a_1,a_2)\in(\new,\retx)$, and we obtain the following:
  \begin{align}
  &g(r) \big[h(e',b+e-E^{tx},\delta^{rx}+2,\delta^{tx}+1,r+1) \nonumber \\& -  h(e',b+e-E^{tx},\delta^{rx}+1,\delta^{tx}+1,r+1)\big] \nonumber \\
  &- \big[h(e',b+e,\delta^{rx}+2,\delta^{tx}+1,r)\nonumber  \\ &- h(e',b+e,\delta^{rx}+1,\delta^{tx}+1,r)\big] \leq 0
  \label{eq:submodular3}
  \end{align}
  \begin{align}
  &g(r) \big[h(e',b+e-E^{tx},\delta^{rx}+2,\delta^{tx}+1,r+1) \nonumber \\& -  h(e',b+e-E^{tx},\delta^{rx}+1,\delta^{tx}+1,r+1)\big] \nonumber \\
  &- g(0) \big[h(e',b+e-E^s-E^{tx},\delta^{rx}+2,1,1) \nonumber \\ &-  h(e',b+e-E^s-E^{tx},\delta^{rx}+1,1,1)\big]\leq 0
  \label{eq:submodular22}
  \end{align}
  Therefore,  \eqref{eq:submodular2}, \eqref{eq:submodular3} and \eqref{eq:submodular22} are the necessary and sufficient conditions for submodularity of $Q$ function.

  First, we note that Eqns.~\eqref{eq:submodular2}, \eqref{eq:submodular3} and \eqref{eq:submodular22} hold with equality for $(\delta^{rx}+1,\delta^{rx})=(M,M-1)$. Then, we show by induction that if \eqref{eq:submodular2}, \eqref{eq:submodular3} and \eqref{eq:submodular22} hold for $(\delta^{rx}+2,\delta^{rx}+1)$ then they hold for $(\delta^{rx}+1,\delta^{rx})$. 
  
  
  First, we check for $(a_1,a_2)=(\idle,\retx)$, and assume that $h$ satisfies \eqref{eq:submodular2}, \eqref{eq:submodular3} and \eqref{eq:submodular22}. We define the related $Q$ functions with optimal actions denoted by $a^*_1$, $a^*_2$, $a^*_3$ and $a^*_4$ such that: 
  \begin{eqnarray}
  \lefteqn{Q(e,b\!-\!E^{tx},\delta^{rx}+1,\delta^{tx}+1,r+1,a^*_1)} \nonumber \\ &\triangleq  h(e,b\!-\!E^{tx},\delta^{rx}\!+\!1,\delta^{tx}\!+\!1,r+1,)\!+\!J^*& \label{eq:1} \\
  \lefteqn{Q(e,b-E^{tx},\delta^{rx},\delta^{tx}+1,r+1,a^*_2)} \nonumber \\ &\triangleq h(e,b-E^{tx},\delta^{rx},\delta^{tx}+1,r+1)+J^*& \label{eq:2}\\
  \lefteqn{Q(e,b,\delta^{rx}+1,\delta^{tx}+1,r,a^*_3)}\nonumber \\ &\triangleq 
  h(e,b,\delta^{rx}+1,\delta^{tx}+1,r) +J^* \label{eq:3} \\
  \lefteqn{Q(e,b,\delta^{rx},\delta^{tx},r,a^*_4)\triangleq 
  h(e,b,\delta^{rx},\delta^{tx},r)+J^*}. \label{eq:4}
  \end{eqnarray}
  


 
   We  need to show that \eqref{eq:submodular3} holds for $(\delta^{rx}+1,\delta^{rx})$, which can be rewritten using \eqref{eq:1}, \eqref{eq:2}, \eqref{eq:3} and  \eqref{eq:4}: 
  \begin{align}
  &g(r) \big[Q(e,b-E^{tx},\delta^{rx}+1,\delta^{tx}+1,r+1,a^*_1) \nonumber \\ & -  Q(e,b-E^{tx},\delta^{rx},\delta^{tx}+1,r+1,a^*_2)\big]  \nonumber \\
  &- \big[Q(e,b,\delta^{rx}+1,\delta^{tx}+1,r,a^*_3) \nonumber \\ &-  Q(e,b,\delta^{rx},\delta^{tx},r,a^*_4)\big] \leq 0
  \label{eq:submodular4}
  \end{align}
  
 We add terms $g(r)[-Q(e,b-E^{tx},\delta^{rx}+1,\delta^{tx}+1,r+1,a^*_2)+Q(e,b-E^{tx},\delta^{rx}+1,\delta^{tx}+1,r+1,a^*_2)]$ and $[Q(e,b,\delta^{rx},\delta^{tx},r,a^*_3)-Q(e,b,\delta^{rx},\delta^{tx},r,a^*_3)]$  to the LHS of \eqref{eq:submodular4} and obtain: 
 \begin{align}
  &g(r) \big[Q(e,b-E^{tx},\delta^{rx}+1,\delta^{tx}+1,r+1,a^*_1)\nonumber \\ &-Q(e,b-E^{tx},\delta^{rx}+1,\delta^{tx}+1,r+1,a^*_2)\nonumber \\ &+Q(e,b-E^{tx},\delta^{rx}+1,\delta^{tx}+1,r+1,a^*_2) \nonumber \\ &-  Q(e,b-E^{tx},\delta^{rx},\delta^{tx}+1,r+1,a^*_2)\big] \nonumber \\
  &- \big[Q(e\!,b\!,\delta^{rx}\!+\!1,\delta^{tx}\!+\!1\!,r\!,a^*_3)\! -\! Q(e,b,\delta^{rx}\!,\delta^{tx}\!,r,a^*_3)\nonumber \\ &+Q(e\!,b\!,\delta^{rx},\delta^{tx},r,a^*_3) -Q(e\!,b\!,\delta^{rx},\delta^{tx},r,a^*_4)\big] \leq 0
  \label{eq:submodular5}
  \end{align}
 $Q(e,b_1,\delta+1,r+1,a^*_1)-Q(e,b_1,\delta+1,r+1,a^*_2)$ is smaller than $0$ from the optimality of $a^*_1$. Similarly, $Q(e,b,\delta^{rx},\delta^{tx},r,a^*_4)- Q(e,b,\delta^{rx},\delta^{tx},r,a^*_3)$  is smaller than $0$ from the optimality of $a^*_4$. Then, we only need to show:
 \begin{align}
     &g(r)[Q(e,b-E^{tx},\delta^{rx}+1,\delta^{tx}+1,r+1,a_2^*) \nonumber \\ &-Q(e,b-E^{tx},\delta^{rx},\delta^{tx}+1,r+1,a_2^*)] \nonumber \\ &-Q(e,b,\delta^{rx}+1,\delta^{tx}+1,r,a_3^*)\nonumber\\&+Q(e,b,\delta^{rx},\delta^{tx},r,a_3^*) \leq 0.
     \label{eq:submodular6}
 \end{align}
 
 The condition in \eqref{eq:submodular6} can be checked for all possible values of $(a_2^*,a_3^*)$ pair: 
 First we investigate for the pair $(a^*_2,a^*_3)=(\idle,\idle)$; that is, 
 \begin{align}
    g(r)\big[& 1+h(e',b-E^{tx}+e,\delta^{rx}+2,\delta^{tx}+1,r+1)\nonumber \\ & -h(e',b-E^{tx}+e,\delta^{rx}+1,\delta^{tx}+1,r+1)\big]\nonumber \\
     & -1 + h(e',b+e,\delta^{rx}+1,\delta^{tx}+1,r) \nonumber \\  &- h(e',b+e,\delta^{rx}+2,\delta^{tx}+1,r) \leq 0.
     \label{eq:submodular10}
 \end{align}
LHS of \eqref{eq:submodular10} is equivalent to LHS of \eqref{eq:submodular3} plus the term $g(r)-1$, which is smaller than and equal to $0$ since \eqref{eq:submodular3} holds  and $g(r) \leq 1$. 
 For pair $(a^*_2,a^*_3)=(\retx,\retx)$; that is,
 \begin{align}
 &g(r)\big[1+g(r+1)h(e',b-2E^{tx}+e,\delta+2,r+2)\nonumber
     \\&-g(r\!+\!1)h(e',b\!-\!2E^{tx}\!+\!e,\delta^{rx}\!+\!1,\delta^{tx}\!+\!1,r\!+\!1)\big]\nonumber\\ &-1-g(r)h(e',b\!-\!E^{tx}\!+\!e,\delta^{rx}+2,\delta^{tx}+1,r+1)\nonumber
     \\&+g(r)h(e',b-E^{tx}+e,\delta^{rx}\!+\!1,\delta^{tx}\!+\!1,r\!+\!1) \leq 0
     \label{eq:submodular11}
\end{align}
     Similarly, $g(r)-1$ is less than 0 and \eqref{eq:submodular11} holds.

For pair $(a^*_2,a^*_3)=(\idle,\retx)$:
 \begin{align}
 &g(r)\big[1+h(e',b-E^{tx}+e,\delta^{rx}+2,\delta^{tx}+1,r+1)\nonumber\\&-h(e',b-E^{tx}+e,\delta^{rx}+1,\delta^{tx}+1,r+1)\big] \nonumber \\
 &-\!1\!-\!g(r)h(e,e',b\!-\!E^{tx}\!+\!e,\delta^{rx}\!+\!2,\delta^{tx}\!+\!1,r\!+1\!)\nonumber
     \\&+g(r)h(e',b\!-\!E^{tx}\!+\!e,\delta^{rx}\!+\!1,\delta^{tx}\!+\!1,r+1) \leq 0,
 \end{align}
which is equal to:
\begin{align}
   g(r)-1 \leq 0.
\end{align}

For pair $(a^*_2,a^*_3)=(\retx,\idle)$:
 \begin{align}
      &g(r)\big[1+g(r+1)h(e',b-2E^{tx}+e,\delta+2,r+2)\nonumber
     \\&-g(r\!+\!1)h(e',b\!-\!2E^{tx}\!+\!e,\delta^{rx}\!+\!1,\delta^{tx}\!+\!1,r\!+\!1)\big]\nonumber\\  & -1 + h(e',b+e,\delta^{rx}+1,\delta^{tx}+1,r) \nonumber \\ &- h(e',b+e,\delta^{rx}+2,\delta^{tx}+1,r) \leq 0.
\end{align}
which is equal to:
\begin{align}
    &g(r)g(r+1)\big[ h(e',b-2E^{tx},\delta+2,r+2)\nonumber \\&-h(e',b-2E^{tx}+e,\delta^{rx}+1,\delta^{tx}+1,r+1)\big] \nonumber \\  &-h(e',b+e,\delta^{rx}+2,\delta^{tx}+1,r)\nonumber \\  &+h(e',b+e,\delta^{rx}+1,\delta^{tx}+1,r) \nonumber \\ &+g(r)-1 + \leq 0.
    \label{c10}
\end{align}

From \eqref{eq:submodular3},  $-h(e',b+e,\delta^{rx}+2,\delta^{tx}+1,r)+h(e',b+e,\delta^{rx}+1,\delta^{tx}+1,r) \leq g(r)(-h(e',b+e-E^{tx},\delta^{rx}+2,\delta^{tx}+1,r+1)+h(e',b+e-E^{tx},\delta+1,r+1))$, and $1-g(r)\leq 0$; thus \eqref{c10} is smaller than
\begin{align}
&g(r)\bigg\{g(r+1)\big[ h(e',b-2E^{tx},\delta+2,r+2)\nonumber \\&-h(e',b-2E^{tx}+e,\delta^{rx}+1,\delta^{tx}+1,r+1)\big] \nonumber \\  &+(-h(e',b+e-E^{tx},\delta^{rx}+2,\delta^{tx}+1,r+1) \nonumber\\ &+h(e',b+e-E^{tx},\delta+1,r+1))\bigg\},
\end{align}
which is smaller than 0, since the expression inside the braces is equivalent to \eqref{eq:submodular3} with $r \to r+1$ and $g(r)\geq 0$.

The same holds for $(a^*_2,a^*_3)=(\retx,\new)$ and $(a^*_2,a^*_3)=(\new,\retx)$. Similar steps could be followed for other $(a_1,a_2)=(\idle,\new)$ and $(a_1,a_2)=(\new,\retx)\}$ pairs and are not included in this paper due to space limitations. 


Thus, the condition is satisfied, i.e., $Q$ function is submodular in $(\delta^{rx},a)$.  From \cite{Topkis1978}, it can be concluded that the status update policy is of threshold-type.

\bibliography{thesis1}

\end{document}